\documentclass[11pt,letterpaper]{article}

\def\showauthornotes{0}
\def\showtableofcontents{1}
\def\showkeys{0}
\def\showdraftbox{0}
\def\showcolorlinks{1}
\def\usemicrotype{1}
\def\showfixme{0}

\def\writemode{0}
\newcommand{\Mat}{\mathsf{Mat}}
\newcommand{\Ten}{\mathsf{Ten}}

\usepackage{etex}

\usepackage[l2tabu, orthodox]{nag}

\usepackage{xspace,enumerate}

\usepackage[dvipsnames]{xcolor}

\usepackage[T1]{fontenc}
\usepackage[full]{textcomp}

\usepackage[american]{babel}

\usepackage{mathtools}

\usepackage{amsthm}

\newtheorem{theorem}{Theorem}[section]
\newtheorem*{theorem*}{Theorem}

\newtheorem*{proposition*}{Proposition}
\newtheorem{lemma}[theorem]{Lemma}
\newtheorem*{lemma*}{Lemma}

\newtheorem*{conjecture*}{Conjecture}
\newtheorem{fact}[theorem]{Fact}
\newtheorem*{fact*}{Fact}
\newtheorem{hypothesis}[theorem]{Hypothesis}
\newtheorem*{hypothesis*}{Hypothesis}

\theoremstyle{definition}
\newtheorem{definition}[theorem]{Definition}

\theoremstyle{remark}

\newtheorem*{claim*}{Claim}

\newtheorem*{remark*}{Remark}

\newtheorem*{observation*}{Observation}

\ifnum\writemode=1
\usepackage[
letterpaper,
top=1.2in,
bottom=1.2in,
left=1in,
right=1in]{geometry}

\fi

\ifnum\writemode=0
\usepackage[
letterpaper,
top=0.7in,
bottom=0.9in,
left=1in,
right=1in]{geometry}
\fi

\usepackage{newpxtext} 
\usepackage{textcomp} 
\usepackage[varg,bigdelims]{newpxmath}
\usepackage[scr=rsfso]{mathalfa}
\usepackage{bm} 
\let\mathbb\varmathbb

\ifnum\showkeys=1
\usepackage[color]{showkeys}
\fi

\ifnum\showcolorlinks=1
\usepackage[
pagebackref,
colorlinks=true,
urlcolor=blue,
linkcolor=blue,
citecolor=OliveGreen,
]{hyperref}
\fi

\ifnum\showcolorlinks=0
\usepackage[
pagebackref,
colorlinks=false,
pdfborder={0 0 0}
]{hyperref}
\fi

\usepackage[capitalise,nameinlink]{cleveref}
\crefname{lemma}{Lemma}{Lemmas}
\crefname{definition}{Definition}{Definitions}

\newcommand{\Sref}[1]{\hyperref[#1]{\S\ref*{#1}}}

\usepackage{nicefrac}

\ifnum\usemicrotype=1
\usepackage{microtype}
\fi

\ifnum\showauthornotes=1
\newcommand{\Authornote}[2]{{\sffamily\small\color{red}{[#1: #2]}}}
\newcommand{\Authornotecolored}[3]{{\sffamily\small\color{#1}{[#2: #3]}}}
\newcommand{\Authorcomment}[2]{{\sffamily\small\color{gray}{[#1: #2]}}}
\newcommand{\Authorstartcomment}[1]{\sffamily\small\color{gray}[#1: }

\newcommand{\Authorfnote}[2]{\footnote{\color{red}{#1: #2}}}
\newcommand{\Authorfixme}[1]{\Authornote{#1}{\textbf{??}}}
\newcommand{\Authormarginmark}[1]{\marginpar{\textcolor{red}{\fbox{\Large #1:!}}}}
\else
\newcommand{\Authornote}[2]{}
\newcommand{\Authornotecolored}[3]{}
\newcommand{\Authorcomment}[2]{}
\newcommand{\Authorstartcomment}[1]{}

\newcommand{\Authorfnote}[2]{}
\newcommand{\Authorfixme}[1]{}
\newcommand{\Authormarginmark}[1]{}
\fi

\newcommand{\Dnote}{\Authornote{D}}

\definecolor{forestgreen(traditional)}{rgb}{0.0, 0.27, 0.13}

\ifnum\showfixme=0

\fi

\usepackage{boxedminipage}

\newcommand{\set}[1]{\{#1\}}

\newcommand{\Esymb}{\mathbb{E}}
\newcommand{\Psymb}{\mathbb{P}}

\DeclareMathOperator*{\E}{\Esymb}

\DeclareMathOperator*{\ProbOp}{\Psymb}

\renewcommand{\Pr}{\ProbOp}

\newcommand{\textparen}[1]{\text{(#1)}}

\ifx\because\undefined
\newcommand{\because}[1]{\textparen{because #1}}
\else
\renewcommand{\because}[1]{\textparen{because #1}}
\fi

\newcommand\bdot\bullet

\ifx\mathds\undefined 

\else
}
\fi

\newcommand{\etal}{et al.\xspace}

\newcommand{\cB}{\mathcal B}

\newcommand{\cS}{\mathcal S}

\newcommand{\cV}{\mathcal V}

\renewcommand{\leq}{\leqslant}

\renewcommand{\geq}{\geqslant}

\ifnum\showdraftbox=1
\newcommand{\draftbox}{\begin{center}
  \fbox{%
    \begin{minipage}{2in}%
      \begin{center}%
          \Large\textsc{Working Draft}\\%
        Please do not distribute%
      \end{center}%
    \end{minipage}%
  }%
\end{center}
\vspace{0.2cm}}
\else
\newcommand{\draftbox}{}
\fi

\let\epsilon=\varepsilon

\numberwithin{equation}{section}

\newcommand\MYcurrentlabel{xxx}

\newcommand{\MYstore}[2]{%
  \global\expandafter \def \csname MYMEMORY #1 \endcsname{#2}%
}

\newcommand{\MYload}[1]{%
  \csname MYMEMORY #1 \endcsname%
}

\newcommand{\MYnewlabel}[1]{%
  \renewcommand\MYcurrentlabel{#1}%
  \MYoldlabel{#1}%
}

\newcommand{\MYdummylabel}[1]{}

\newcommand{\torestate}[1]{%
  \let\MYoldlabel\label%
  \let\label\MYnewlabel%
  #1%
  \MYstore{\MYcurrentlabel}{#1}%
  \let\label\MYoldlabel%
}

\newcommand{\restatetheorem}[1]{%
  \let\MYoldlabel\label
  \let\label\MYdummylabel
  \begin{theorem*}[Restatement of \prettyref{#1}]
    \MYload{#1}
  \end{theorem*}
  \let\label\MYoldlabel
}

\newcommand{\restatelemma}[1]{%
  \let\MYoldlabel\label
  \let\label\MYdummylabel
  \begin{lemma*}[Restatement of \prettyref{#1}]
    \MYload{#1}
  \end{lemma*}
  \let\label\MYoldlabel
}

\newcommand{\restateprop}[1]{%
  \let\MYoldlabel\label
  \let\label\MYdummylabel
  \begin{proposition*}[Restatement of \prettyref{#1}]
    \MYload{#1}
  \end{proposition*}
  \let\label\MYoldlabel
}

\newcommand{\restatefact}[1]{%
  \let\MYoldlabel\label
  \let\label\MYdummylabel
  \begin{fact*}[Restatement of \prettyref{#1}]
    \MYload{#1}
  \end{fact*}
  \let\label\MYoldlabel
}

\newcommand{\restate}[1]{%
  \let\MYoldlabel\label
  \let\label\MYdummylabel
  \MYload{#1}
  \let\label\MYoldlabel
}

\newcommand{\addreferencesection}{
  \phantomsection
  \addcontentsline{toc}{section}{References}
}

\let\origparagraph\paragraph
\renewcommand{\paragraph}[1]{\origparagraph{#1.}}

\allowdisplaybreaks

\usepackage{paralist}

\usepackage{comment}

\usepackage{braket}

\usepackage{relsize}
\usepackage[font=footnotesize]{caption}
\usepackage{appendix}

\DeclareMathOperator{\Span}{Span}

\DeclareMathOperator{\F}{\mathbb{F}}

\DeclareUrlCommand\email{}

\renewcommand{\cal}{\mathcal}

\usepackage{mdframed}
\usepackage{lipsum}
\usepackage{amsmath}

\newcounter{boxeddenv}
\newenvironment{boxedd}[1]
  {
 \mdfsetup{
    frametitle={ \colorbox{white}{\space Test \theboxeddenv: \space #1\space}},
    innertopmargin=10pt,
    frametitleaboveskip=-\ht\strutbox,
    frametitlealignment=\center
    }
  \begin{mdframed} %
  }
  {\refstepcounter{boxeddenv} \end{mdframed}
   
   }

\title{Small-Set Expansion in Shortcode Graph and the 2-to-2 Conjecture}

\author{
Boaz Barak \thanks{Harvard University, \protect \email{b@boazbarak.org}. Supported by NSF awards CCF 1565264 and CNS 1618026, and the Simons Foundation. Part of the work done while the author visited Weizmann Institute during Spring 2017.}
\and
Pravesh K. Kothari \thanks{Princeton University and IAS \protect \email{kothari@cs.princeton.edu}. Work done while the author visited Weizmann Institute in March 2017.
}
\and
David Steurer \thanks{ETH Zurich \protect \email{dsteurer@cs.cornell.edu}. Work done while the author was a member of the Institute for Advanced Study, Princeton.}
}

\begin{document}

\maketitle
 \draftbox
\thispagestyle{empty}
\begin{abstract}
Dinur, Khot, Kindler, Minzer and Safra~\cite{DBLP:journals/eccc/DinurKKMS16} recently showed that the (imperfect completeness variant of) Khot's 2 to 2 games conjecture follows from  a  combinatorial hypothesis about the soundness of a certain ``Grassmanian agreement tester''.
In this work, we show that  hypothesis of Dinur et al  follows from a conjecture we call the ``Inverse Shortcode Hypothesis'' characterizing the non-expanding sets of the degree-two shortcode graph. We also show the latter conjecture is equivalent to a characterization of the non-expanding sets in the Grassman graph, as hypothesized by a follow-up paper of Dinur \etal ~\cite{DBLP:journals/eccc/DinurKKMS17}.

Following our work, Khot, Minzer and Safra~\cite{KMS} proved the ``Inverse Shortcode Hypothesis''. Combining their proof with our result  and the reduction of \cite{DBLP:journals/eccc/DinurKKMS16}, completes the proof of the 2 to 2 conjecture with imperfect completeness.
Moreover, we believe that the shortcode graph provides a useful view  of both the hypothesis and the reduction, and might be useful in extending it further.
\end{abstract}

\clearpage

\ifnum\showtableofcontents=0
{
\tableofcontents
\thispagestyle{empty}
 }
\fi

\clearpage

\setcounter{page}{1}

\newcommand{\Proj}{\mathsf{Proj}}
\newcommand{\bB}{\mathcal{B}}
\newcommand{\Lin}{\mathsf{LIN}}
\newcommand{\grassmann}{\mathcal{G}}

\section{Introduction}

In~\cite{Khot2}, Subhash Khot put forward a family of conjectures known as the ``$d$-to-$d$ games conjectures''. A binary constraint $P(x_1, x_2)$ where $x_i$s take values in alphabet $\Sigma$ is said to be \emph{d-to-d} if for every value to $x_1$, there are exactly $d$ values for $x_2$ that satisfy $P$ and vice-versa. For any $d$, the ``$d$-to-$d$ games conjecture'' roughly says that for every $\epsilon>0$, there is some finite  alphabet $\Sigma$ such that it is NP-hard to distinguish, given a constraint satisfaction problem with d-to-d constraints, whether it is possible to satisfy at least $1-\epsilon$ fraction of the constraints, or if every assignment satisfies at most $\epsilon$ fraction of the constraints. 
\footnote{For $d>1$, the conjectures are often stated in their \textit{perfect completeness variant}, where we replace $1-\epsilon$ with $1$ in the first case. In this work (as well as all the line of works following \cite{DBLP:conf/stoc/KhotMS17}), we refer to the imperfect completeness version as stated above.}
The case of $d=1$ corresponds to the more famous \textit{Unique Games Conjecture}, but until recently there was no constant $d$ for which the corresponding d-to-d conjecture was known to be true.

Dinur, Khot, Kindler, Minzer, and Safra~\cite{DBLP:journals/eccc/DinurKKMS16}, building on ideas of Khot, Minzer and Safra \cite{DBLP:conf/stoc/KhotMS17}, recently initiated an approach towards proving the $2$-to-$2$ conjecture, based on a certain combinatorial hypothesis positing the soundness of the ``Grassmann agreement test''. 

In this work we show that their hypothesis follows from a certain natural hypothesis characterizing the structure of non-expanding sets in the degree two shortcode graph~\cite{MR3416138-Barak15}. Following our work, Khot, Minzer and Safra~\cite{KMS} proved the latter hypothesis thus completing the proof of the 2-to-2 games conjecture.
This has several implications to hardness of approximation including improving on the NP-hardness of approximation for Vertex Cover along with a host of other improved NP-hardness results. Perhaps more importantly, this also gives a strong evidence for the truth of the Unique Games Conjecture itself. We defer to \cite{DBLP:journals/eccc/DinurKKMS16,DBLP:journals/eccc/DinurKKMS17,KMS} for a detailed discussion on the importance of the 2-to-2 games conjecture, as well as the reduction of this conjecture to showing the soundness of the Grassmann agreement tester.

\subsection{Our Results}

Our main result reduces the task of proving the ``Grassmann agreement hypothesis'' of Dinur et al \cite[Hypothesis 3.6]{DBLP:journals/eccc/DinurKKMS16} to characterizing the structure of non-expanding sets in the associated Grassmann graph.

\begin{itemize}

\item We show that the Grassmann agreement hypothesis  \cite[Hypothesis 3.6]{DBLP:journals/eccc/DinurKKMS16} follows from the Grasmann Expansion Hypothesis \cite[Hypothesis 1.7]{DBLP:journals/eccc/DinurKKMS17}.  

\item We describe the related Shortcode test and the associated agreement and expansion hypothesis and relate them to the Grassmann versions above.



\end{itemize}

The above, combined with the work of \cite{DBLP:journals/eccc/DinurKKMS16,KMS}, suffices to prove the 2-to-2 conjecture.
However we note that it is possible to directly obtain a proof of the 2-to-2 conjecture (see the recent exposition at \cite{dkkmsnotes}) using the ``Inverse Shortcode Hypothesis'' without going through the Grassmann graph at all. We think the shortcode view provides a  natural way to understand the reduction and suggests potential extensions, see Section~\ref{sec:discuss}.

\subsection{Grassmann Graph and DKKMS Consistency Test}

To state our results formally, we need to define the Grassman and shortcode graphs, which we now do.
The Grassmann graph $\grassmann(\ell,n)$ with parameters $\ell,n$ has vertices given by all $\ell$-dimensional subspaces (denoted by $\cV_{\ell}$) of $\F_2^n.$
Two subspaces $V,V'$ of $\F_2^n$ have an edge between them if $\dim(V \cap V') = \ell-1$.

Let $\Lin(\F_2^n)$ be the set of all linear functions $\F_2^n \rightarrow \F_2.$ For every $f \in \Lin(F_2^n)$, let $F_f$ be the map that assign to every $V \in \cV_{\ell}$, $F_f(V) = f_{\mid V}$ the restriction of the linear function $f$ to the subspace $V$.
Let $\Lin(\ell,n)=\set{F_f \mid f\in \Lin(\F_2^n)}$ be the set of all such maps.

The Grassmann Consistency test is a two-query test for $\Lin(\ell,n)$ described below:

\begin{boxedd}{Grassmann Consistency Test}
  \begin{description}
    \item[Given:]
      a map $F$ from $\cV_{\ell} \rightarrow \Lin(F_2^{\ell})$ that maps any $V \in \cV_{\ell}$ to $F(V)$ a linear function on $V$.
    \item[Operation:]\mbox{}
      \begin{enumerate}
      	\item Pick an edge $(V,V')$ of $\grassmann(\ell,n)$ uniformly at random.
		\item Receive $F(V), F(V') \in \Lin(\ell,n)$. 
		\item Accept if $F(V)_{V \cap V'} = F(V')_{V \cap V'}$ otherwise reject.
      \end{enumerate}
    \end{description}
\end{boxedd}



It is easy to see the following completeness of the Grassmann graph test. 
\begin{fact}[Completeness]
Suppose $F \in \Lin(\ell,n).$ Then,  $F$ passes the Grassman Consistency test with probability $1$.
\end{fact}

The DKKMS hypothesis conjectures a precise version of soundness of the Grassmann Consistency Test. 

\begin{hypothesis}[DKKMS Soundness Hypothesis] \label{hypothesis:soundness}
For every $\delta > 0$, there exists $\epsilon > 0$, and an integer $r > 0$ such that following holds for sufficiently large $n \gg \ell.$

Let $F:\cV_{\ell} \rightarrow \Lin(\F_2^{\ell})$ such that $\Pr_{(V,V') \sim \grassmann(\ell,n)}[ F(V)_{V \cap V'} = F(V')_{V \cap V'}] \geq \delta.$ 
Then, there exist subspaces $Q, W \subseteq \F_2^{n}$ of dimensions $r$ and $n-r$ respectively and a $f \in \Lin(\F_2^{n})$ such that 
\[
\Pr_{V \sim \cV_{\ell}, Q \subseteq V \subseteq W} [ F(V) = f_V] \geq \epsilon.
\]
\end{hypothesis}

\subsection{Shortcode Graph and Consistency Test}

We now define the closely related \emph{Degree 2 Shortcode} graph and a immediate analog of the Grassmann consistency test on this graph. For parameters $\ell,n$ as before, the vertices of the degree 2 shortcode graph $\cS_{\ell,n}$ are elements of $\Mat_{\ell,n}$, that is, all matrices on $\F_2$ with dimensions $\ell \times n$. Two vertices $M_1$ and $M_2$ have an edge between them if $M_1-M_2$ is a rank 1 matrix over the field $\F_2$. The 2 query codeword test on this graph is entirely analogous to the one above for the Grassmann graph: 

\begin{boxedd}{Degree 2 Shortcode Consistency Test}
  \begin{description}
    \item[Given:] a map $F$ from $\Mat_{\ell,n} \rightarrow \F_2^{\ell}$.
    \item[Operation:]\mbox{}
      \begin{enumerate}
\item Pick $M_1 \sim \Mat_{\ell,n}$ and a rank 1 matrix $ab^{\top}$ for vectors $a \in \F_2^{\ell}$, $b \in \F_2^{n}$ all uniformly at random from their respective domains. Let $M_2 = M_1 + ab^{\top}$.
\item Receive $F(M_1),F(M_2) \in \F_2^{\ell}$.
\item Accept if $F(M_2) \in \{ F(M_1), F(M_1) + a\}$. 
      \end{enumerate}
    \end{description}
    \label{test:deg-2-shortcode}
\end{boxedd}

Just as the Grassmann consistency test, the above shortcode consistency test is "2-to-2" constraint and the following completeness is easy to establish.

\Dnote{linear vs affine linear?}

\begin{fact}[Completeness]
Let $f:\F_2^n \rightarrow \F_2$ be any affine linear function. Let $F = F_f:\Mat_{\ell,n} \rightarrow \F_2^{\ell}$ be the map that evaluates $f$ on each row the input matrix. Then, $F$ passes the shortcode consistency test with probability 1.
\end{fact}

The analogous soundness hypothesis can now be stated as:
\begin{hypothesis}[Degree 2 Shortcode Soundness Hypothesis] \label{hypothesis:soundness-shortode}
For every $\delta > 0$, there exists $\epsilon > 0$, and an integer $r > 0$ such that following holds for sufficiently large $n \gg \ell.$

Let $F:\Mat_{\ell,n} \rightarrow \F_2^{\ell}$ such that $\Pr_{M \sim \Mat_{\ell,n}, a \sim \F_2^{\ell}, b \sim \F_2^{n}} [F(M + ab^{\top}) \in \{ F(M), F(M)+a\}] \geq \delta.$ 
Then, there exists linear constraints $(q_i,t_i)$ and $(r_i,s_i)$ for $i \leq r$ and a $z \in \F_2^{n}, u \in \F_2^{\ell}$ such that 
\[
\Pr_{M \sim \Mat_{\ell,n}} [ F(M) = Mz + u \mid Mq_i = t_i, r_i^{\top}M = s_i \text { } \forall i \leq r] \geq \epsilon.
\]
\end{hypothesis}

\subsection{Soundness vs Small-Set Expansion in Grasmann/Shortcode Graphs}
Recall that for a regular graph $G$, the expansion of a set $S$ of vertices is the probability, that a random walk beginning at a uniformly random vertex in $S$ steps out of $S$. That is, $\Phi_G(S) = \Pr_{v \sim S, v' \sim v}[ v' \not \in S].$

The DKKMS Soundness Hypothesis implies a natural characterization small non-expanding sets in the $\grassmann(\ell,n)$ noted below as Hypothesis \ref{hypothesis:expansion}. Similarly, the degree 2 shortcode soundness hypothesis implies a natural characterization of non-expanding sets in $\cS_{\ell,n}$. We include a brief overview of the argument here and refer the reader to the more extensive commentary in Section 1.3 of \cite{DBLP:journals/eccc/DinurKKMS16} for further details. 

Suppose $A_1, A_2, \ldots, A_r$ are ``non-expanding'' sets that cover a constant fraction of vertices in $\grassmann(\ell,n).$ We construct a labeling strategy $F$ by choosing $r$ uniformly random linear functions $f_i:\F_2^n \rightarrow \F_2$ and setting $F(V) = f_i$ if $V \sim A_i$ and $F(V)$ is a random linear function otherwise. Clearly, $F$ doesn't agree with a single linear function on significantly more than $1/r$ fraction of the vertices in $\cV_{\ell}.$ On the other hand, if $A_i$s are sufficiently non-expanding, then, a random edge will lie inside one of the $A_i$s with a non-trivially large probability and thus $F$ will satisfy the Grassmann consistency test. In this, case, we will hope that there are subspaces $Q,W$ of constant dimension and co-dimension, respectively such that restricting to subspaces $V \in \cV_{\ell}(Q,W)$ (where $\cV_{\ell}(Q,W)$ is the subset $V \in \cV_{\ell}$ such that $V \subseteq W$) implies that $F(V) = f_{V}$ for some fixed global linear function $f$. This can happen in the above example for $F$ only if there are $Q,W$ as above such that one of the $\frac{|A_i \cap \cV_{\ell}(Q,W)|}{|\cV_{\ell}(Q,W)}$ is $\Omega(1)$ (i.e. independent of $\ell,n$). Thus, Hypothesis \ref{hypothesis:soundness} forces that the non-expanding sets $A_i$ to be ``structured'' (in the sense of having a large density inside $\cV_{\ell}(Q,W)$ for some $Q,W$ of constant dimension and co-dimension, respectively.) This can be interpreted as saying that the non-expansion of any set of vertices in $\grassmann(\ell,n)$ can be ``explained'' away by a more than typical density in one of the canonical non-expanding sets (i.e., those that contain a subspace $Q$ and are contained inside a subspace $W$ of constant dimension and co-dimension, respectively.)

To formally state the Grassmann Expansion Hypothesis, we define the special non-expanding sets (referred to as ``zoom-in'' and ``zoom-outs'' in \cite{DBLP:journals/eccc/DinurKKMS17}):

\begin{definition}[Nice Sets in Grassmann Graph]
A subset $S \subseteq \cV_{\ell}$ of vertices in $\grassmann(\ell,n)$ is said to be $r$-nice if there are subspace $Q,W$ of $\F_2^{n}$ of dimension and co-dimension $r_1, r_2$ respectively such that $r_1 + r_2 = r$ and $S = \{ V \subseteq \cV_{\ell} \mid Q \subseteq V \subseteq W\}.$
\end{definition}

\begin{hypothesis}[Grassmann Expansion Hypothesis]
For every $\eta > 0$, there exists $\delta, r$ depending only on $\eta$ such that if $S \subseteq \cV_{\ell}$ satisfies $\Phi_{\grassmann(\ell,n)}(S) < \eta$, then, there are subspaces $Q,W$ over $\F_2^n$ of dimension and co-dimension $r_1, r_2$ satisfying $r_1 + r_2 \leq r$ respectively, such that $\Pr_{V: Q \subseteq V  \subseteq W} [ V\in S] \geq \delta.$ \label{hypothesis:expansion} 
\end{hypothesis}

Analogously, we can define nice sets in the degree 2 shortcode graph and state the expansion hypothesis. We call $Q$, a \emph{right} affine subspace of matrices in $\Mat_{\ell,n}$ if there are pairs $(q_i,t_i)$ and every $M \in Q$ satisfies $Mq_i = t_i$. We define a \emph{left} affine subspace analogously.

\begin{definition}[Nice Sets in Degree 2 Shortcode Graph]
A subset $S \subseteq \cS_{\ell,n}$ is said to be $r$-nice if it is an intersection of a left and right affine subspace in $\Mat_{\ell,n}$ with sum of the dimensions $r$. 
\end{definition}

\begin{hypothesis}[Inverse Shortcode Hypothesis]
For every $\eta > 0$, there exist $\delta, r$ depending only on $\eta$ such that for every subset $S \subseteq \Mat_{\ell,n}$, if $\Pr_{M \sim S, a \sim \F_2^{\ell}, b \sim \F_2^{n}} [ M + ab^{\top} \in S] \geq \eta$, then, there exists an $r$-nice set $\cal{T} \subseteq \cS_{\ell,n}$ such that $|S \cap \cal{T}| \geq \delta |\cal{T}|$.  \label{conj:inv-short-code}
\end{hypothesis}

While Hypotheses \ref{hypothesis:soundness} and \ref{hypothesis:soundness-shortode} posit soundness of a specific ``code-word consistency'' test associated with the Grassmann/Shortcode graphs, Hypotheses \ref{hypothesis:expansion} and \ref{conj:inv-short-code} ask for a purely graph theoretic property: a characterization of non-expanding sets in $\grassmann(\ell,n)$ and $\cS_{\ell,n}$. 
As such, it appears easier to attack and \cite{DBLP:journals/eccc/DinurKKMS16} thus suggested understanding the structure of non-expanding sets in $\grassmann(\ell,n)$ as a natural first step. As we show in this note, proving Hypothesis \ref{conj:inv-short-code} is in fact enough to show Hypothesis \ref{hypothesis:soundness}. In a follow up work \cite{KMS}, this result was used in  to complete the proof of the DKKMS soundness hypothesis. 


\subsection{Our Results}
We are now ready to state our main results formally.

First, we show that the soundness of the shortcode consistency test follows from the expansion hypothesis for the shortcode graph.
\begin{theorem} \label{thm:soundness=expansion}
The degree 2 Shortcode Expansion Hypothesis \ref{conj:inv-short-code} implies the Degree 2 Shortcode Soundness Hypothesis \ref{hypothesis:soundness-shortode}.

\end{theorem}

Second, we show that the soundness hypothesis for the shortcode consistency test implies the soundness hypothesis for the Grassmann consistency test. This reduces the DKKMS soundness hypothesis to establishing the expansion hypothesis for the Shortcode graph.

\begin{theorem} \label{cor:soundness-shortcode-vs-grassmann}
The degree 2 Shortcode Soundness Hypothesis implies the Grassmann Soundness Hypothesis \ref{hypothesis:soundness}.
\end{theorem}

Finally, we relate the expansion hypothesis of the Grassmann graph to the expansion hypothesis for the degree 2 shortcode graph. 

\begin{theorem} \label{thm:expansion-equivalent}

The Grassmann Expansion Hypothesis (Hypothesis \ref{hypothesis:expansion}) is equivalent to the Inverse Shortcode Hypothesis (Hypothesis \ref{conj:inv-short-code}).
\end{theorem}


\subsection{Discussion}
\label{sec:discuss}

Working with the shortcode consistency test (and consequently, the shortcode expansion hypothesis) makes an approach to proving Hypothesis \ref{hypothesis:soundness} somewhat more tractable. This is because unlike the Grassmann graph, Degree 2 shortcode graph is a Cayley graph on $\F_2^{\ell n}$ (isomorphic to $\Mat_{\ell,n}$) under the group operation of $\F_2$-addition with the set of all rank 1 matrices forming the set of generators. Thus studying expansion of sets of vertices can be approached via powerful methods from Fourier analysis. Indeed, this is the route taken by the recent breakthrough  \cite{KMS} that proves the shortcode expansion hypothesis and completes the proof of the 2-to-2 games conjecture (with imperfect completeness).

Perhaps equally importantly, the shortcode consistency test suggests immediate extensions (\emph{higher degree shortcode graphs}) that provide a natural path to proving the Unique Games Conjecture. We discuss this approach here.

First, the Grassmann/shortcode consistency tests as stated above are ``2-to-2'' tests. That is, for any reply for the first query, there are two admissible replies for the other query. However, it is simple to modify the tests and make them \emph{unique} or ``1-to-1'' at the cost of making the completeness $1/2$ instead of $1$. For concreteness, we describe this simple modification below.  

\begin{boxedd}{Unique Degree 2 Shortcode Consistency Test}
  \begin{description}
    \item[Given:] a map $F$ from $\Mat_{\ell,n} \rightarrow \F_2^{\ell}$.
    \item[Operation:]\mbox{}
      \begin{enumerate}
\item Pick $M_1 \sim \Mat_{\ell,n}$ and a rank 1 matrix $ab^{\top}$ for vectors $a \in \F_2^{\ell}$, $b \in \F_2^{n}$ all uniformly at random from their respective domains. Let $M_2 = M_1 + ab^{\top}$.
\item Receive $F(M_1),F(M_2) \in \F_2^{\ell}$.
\item Accept if $F(M_2) = F(M_1)$.
      \end{enumerate}
    \end{description}
    \label{test:unique-shortcode}
\end{boxedd}
\begin{boxedd}{Unique Degree 3 Shortcode Consistency Test}
  \begin{description}
    \item[Given:] a map $F$ from $\Ten_{\ell, m,n} \rightarrow \F_2^{\ell}$.
    \item[Operation:]\mbox{}
      \begin{enumerate}
\item Pick $T_1 \sim \Ten_{\ell,m,n}$ and a rank 1 tensor $a \otimes b \otimes c$ for vectors $a \in \F_2^{\ell}$, $b \in \F_2^{m}$ and $c \in \F_2^{n}$ all uniformly at random from their respective domains. Let $T_2 = T_1 + a \otimes b \otimes c$.
\item Receive $F(T_1),F(T_2) \in \F_2^{\ell}$.
\item Accept if $F(T_2) = F(T_1)$.
      \end{enumerate}
    \end{description}
\end{boxedd}

It is easy to check that the any strategy that passes the 2-to-2 test can be modified to obtain a success probability of $1/2$ in passing the ``unique'' test above (see proof of Lemma \ref{lem:uniquifying} below). This is one of the several ways that the NP hardness of ``2-to-2'' games implies the NP hardness of $(1/2,\epsilon)$-unique games - that is, distingushing between instances where at least $1/2$ the constraints are satisfiable from those where at most $\epsilon$ fraction of constraints are satisfiable.

A natural strategy, thus, to try to show NP hardness of $(1-\epsilon,\epsilon)$-unique games is to use some variant of the shortcode consistency test above that has completeness $1-\epsilon$ instead of $1/2$. Indeed, the degree 2 shortcode consistency test suggests natural analogs with higher completeness - by moving to higher degree shortcode graphs. For concreteness, consider the following test on degree 3 shortcode graphs, where it is easy to argue a completeness of $3/4$.

Let $\Ten_{\ell, m, n}$ be the set of all $\ell \times m \times n$ tensors over $\F_2$. Recall that a rank 1 tensor is defined by 3 vectors $a \in \F_2^{\ell}$, $b \in \F_2^{m}$ and $c \in \F_2^{n}$ and can be written as $a \otimes b \otimes c$.

To see why there's a natural analog of the strategy in case of the degree 2 shortcode consistency test that gives a completeness of $3/4$, we show:

\begin{lemma}[Completeness]
Let $y\in \F_2^{m}$ and $z \in \F_2^{n}$. Let $F_f:\Ten_{\ell,m,n} \rightarrow \F_2^{\ell}$ be the map that assigns to any tensor $T$, the value $F(T)_i = \sum_{j,k} T(i,j,k) y_j z_k$. Then, $F_f$ passes the test with probability $3/4$.
\end{lemma}

\begin{proof}
Let $T,T'$ be such that $T-T'$ is rank 1 tensor. Then, $F_f$ passes the test only if $F_f(T-T') = 0$. If $T-T'= a \otimes b \otimes c$, then $F_f(T-T') = \langle b, y \rangle \cdot \langle c, z \rangle a$. Since $b,c$ are independently chosen in the test, the probability that $F_f(T-T') = 0$ is $3/4$.  
\end{proof}

Thus, the degree 3 shortcode consistency test gives a natural analog of the degree 2 shortcode consistency test with higher completeness. Indeed, degree r version gives a test with completeness of $1-2^{-r}$ as expected. One can also frame expansion hypotheses similar to the ones for the degree 2 case that posit a characterization of the non-expanding sets in higher degree shortcode graphs. 

While our current efforts to compose this test with the ``outer-PCP'' in order to get a reduction to Unique Games problem (with higher completeness) have not succeeded, it seems a natural avenue for launching an attack on the UGC.\footnote{There are indeed very serious obstacles that must be overcome before carrying this out. Specifically, the reduction of \cite{DBLP:journals/eccc/DinurKKMS16} uses a careful interplay between \textit{smoothness} properties of the outer PCP and \textit{efficiency} or ``blow up'' properties of the test (i.e., the number of potential queries by the verifier as a function of the number of honest strategies). The tensor based test has too much of a blowup to be able to be simply ``plugged in'' in the outer PCP used by~\cite{DBLP:journals/eccc/DinurKKMS16}.  }

\section{Small-Set-Expansion vs Soundness}
In this section, we establish that the inverse shortcode hypothesis (Hypothesis \ref{conj:inv-short-code}) implies the soundness of the degree 2 shortcode consistency test \ref{hypothesis:soundness-shortode}. 

\paragraph{From 2-to-2 to Unique Tests}
For the sake of exposition, it will be easier to work with Test \ref{test:unique-shortcode}, the ``unique'' version of the degree 2 shortcode consistency test. Thus, we restate the soundness hypothesis for Test \ref{test:unique-shortcode} and show that it is enough to establish Hypothesis  \ref{hypothesis:soundness-shortode}.
\begin{hypothesis}[Soundness of Test \ref{test:unique-shortcode}] \label{hypothesis:soundness-unique}
For every $\eta > 0$, there exists $\delta > 0$, and an integer $r > 0$ such that following holds for sufficiently large $n \gg \ell.$

Let $F:\Mat_{\ell,n} \rightarrow \F_2^{\ell}$ such that $\Pr_{M \sim \Mat_{\ell,n}, a \sim \F_2^{\ell}, b \sim \F_2^{n}} [F(M + ab^{\top}) =F(M)] \geq \eta.$ 
Then, there exists linear constraints $(q_i,t_i)$ and $(r_i,s_i)$ for $i \leq r$ and a $z \in \F_2^{n}, u \in \F_2^{\ell}$ such that 
\[
\Pr_{M \sim \Mat_{\ell,n}} [ F(M) = Mz + u \mid Mq_i = t_i, r_i^{\top}M = s_i \text { } \forall i \leq r] \geq \delta.
\]
\end{hypothesis}

We first show that Hypothesis \ref{hypothesis:soundness-unique} implies Hypothesis \ref{hypothesis:soundness-shortode}.
\begin{lemma} \label{lem:uniquifying}
Hypothesis \ref{hypothesis:soundness-unique} implies Hypothesis \ref{hypothesis:soundness-shortode}.
\end{lemma}

\begin{proof}
Let $F$ be the labeling strategy for Test \ref{test:deg-2-shortcode}. We will first obtain a good labeling strategy for Test \ref{test:unique-shortcode} by modifying $F$ slightly. 

Choose $h$ uniformly at random from $\F_2^{n}$. For any $M \in \Mat_{\ell,n}$, let $G(M) = F(M) + Mh$. 
We claim that if $F$ passes the Test \ref{test:deg-2-shortcode} with probability $\eta$, then $G$ passes Test \ref{test:unique-shortcode} with probability at least $\eta/2$.

To see this, take any $M,M'$ such that $M \sim M'$ in $\cS_{\ell,n}$. That is, $M-M' = ab^{\top}$ for vectors $a,b$. We will argue that $G(M) = G(M')$ with probability $1/2$. This will imply that in expectation over the choice of $h$, $G$ satisfies at least $1/2$ the constraints satisfied by $F$ in Test \ref{test:deg-2-shortcode} completing the proof.

This is simple to see: since $F$ passes the test, $F(M) = F(M')$ or $F(M) -F(M') = a$. WLOG, say the first happens. Observe that $G$ passes the unique test on $M,M'$ if $F(M) + Mh = F(M') + M'h$ or $F(M) -F(M') = (M-M')h = \langle b,h \rangle a$. Since $F(M) = F(M')$, $G$ thus passes if $\langle b, h \rangle = 0$ which happens with probability $1/2$. 





\end{proof}

\paragraph{Expansion to Soundness} We will now show that Hypothesis \ref{conj:inv-short-code} implies Hypothesis \ref{hypothesis:soundness-unique}. This completes the proof of Theorem \ref{thm:soundness=expansion}. A similar argument can be used to directly establish that Hypothesis \ref{hypothesis:expansion} implies Hypothesis \ref{hypothesis:soundness}. We do not include it here explicitly. Instead, we relate the expansion and soundness hypothesis for the degree 2 shortcode test to the analogs for the Grassmann test as we believe this could shed light on showing expansion hypotheses for higher degree shortcode tests discussed in the next section.

\begin{lemma}
Hypothesis \ref{conj:inv-short-code} implies Hypothesis \ref{hypothesis:soundness-unique}
\end{lemma}
\begin{proof}
Let $F$ be the labeling function as in the assumption in Hypothesis \ref{hypothesis:soundness-unique}. Then, we know that $\Pr_{M \sim \Mat_{\ell,n}, a \sim \F_2^{\ell}, b \sim \F_2^{n}} [ F(M) = F(M+ab^{\top})] \geq \eta.$ For any $z \in \{0,1\}^{\ell}$, let $S_z$ be the set of all matrices $M$ with $F(M) = z$. Then, by an averaging argument, there must be a $z \in \{0,1\}^{\ell}$ such that $\Pr_{M \sim S_z, a \sim \F_2^{\ell}, b \sim \F_2^{n}} [ M+ab^{\top} \in S_z] \geq \eta.$

Apply Hypothesis \ref{conj:inv-short-code} to $S_z$ to obtain $r$-nice subset $Q$ of $\Mat_{\ell,n}$ such that $|Q \cap S_z| \geq \delta |Q|$. Let $Mq = t$ be a affine constraint satisfied by every $M \in Q$. Consider the affine linear strategy $H:\Mat_{\ell,n} \rightarrow \F_2^{\ell}$ that maps any $M$ to $H(M) = Mq + t + z$. Observe that for every $M \in Q$, $H(M) = z$ by this choice. As a result, when $M \sim M'$ are such that $M,M' \subseteq Q$, $\Pr[ H(M) = H(M')] \geq \delta$. Thus, $H$ is the ``decoded'' strategy that satisfies the requirements of Hypothesis \ref{hypothesis:soundness-unique} as required. This completes the proof.



\end{proof}

\section{Relating Grassmann Graphs to Degree 2 Shortcode Graphs}
In this section, we show a formal relationship between the Grassmann and the degree Shortcode tests. In particular, we will prove Theorems \ref{cor:soundness-shortcode-vs-grassmann} and \ref{thm:expansion-equivalent}.

\subsection{A homomorphism from $\grassmann(\ell,n)$ into $\cS_{\ell,n}$}
Key to the relationship between the two tests is an embedding of the degree 2 shortcode graph $\cS_{\ell,n}$ into $\Mat_{\ell,n-\ell}$. We describe this embedding first. As justified in the previous section, it is without loss of generality to work with the ``unique'' versions of both the tests.

To describe the above embedding, we need the notion of \emph{projection} of a subspace of $\F_2^{n}$ to a set of coordinates. 

\begin{definition}[Projection of a Subspace]
Given a subspace $V \subseteq \F_2^n$, the projection of $V$ to a set of coordinates $S \subseteq [n]$, written as $\Proj_S(V)$ is the subspace of $\F_2^{|S|}$ defined by taking the vectors obtained by keeping only the coordinates indexed by $S$ for every vector $v \in V.$
\end{definition}

Let $\bB \subseteq \F_2^n$ be the set $n$-tuples of linearly independent elements of $\F_2^n$, i.e. each $B \in \bB $ forms a basis for the vector space $\F_2^n$. We will use $B_0$ to denote the standard basis $\{e_1, e_2, \ldots, e_n\}.$ 

We will now describe a class of graph homomorphisms from $\grassmann(\ell,n)$ into $\cS_{\ell,n-\ell}$. Each element of this class can be described by a basis $B$ of $\F_2^n$.

For each basis $B \in \bB$, let $\cV_{\ell}(B) \subseteq \cV_{\ell}$ be the set of all subspaces $V \in \cV_{\ell}$ such that the projection of $V$ to the first $\ell$ coordinates when written w.r.t. the basis $B$ is full-dimensional. Our embedding will map each element of $\cV_{\ell}(B)$ into a distinct element of $\Mat_{\ell,n}$ such that the edge structure within $\cV_{\ell}(B)$ in $\grassmann (\ell,n)$ is preserved under this embedding.

\begin{definition}[Homomorphism from $\grassmann(\ell,n)$ into $\cS_{\ell,n}$ ]
Let $\phi = \phi_B: \mathcal{V}_{\ell}(B)\rightarrow Mat_{\ell,n-\ell}$ be defined as follows. Write every vector in the $B$-basis.
For any $V \in \mathcal{V}_{\ell}(B)$ and for $1 \leq i \leq \ell$, let $v_i$ be the unique vector in $V$ such that $Proj_{[\ell]}(v_i) = e_i \in \F_2^{\ell}$. We call $v_1, v_2, \ldots,v_{\ell}$ to be the \emph{canonical basis} for $V$.

Define $\phi(V)$ to be the $\ell \times (n-\ell)$ matrix with the $i^{th}$ row given by the projection of $v_i$ on the last $(n-\ell)$ coordinates for each $1 \leq i \leq \ell$.  When the basis $B$ is clear from the context,  we will omit the subscript and write $\phi$.
\end{definition}
It is easy to confirm that $\phi$ is a bijection from $\cV_{\ell}(B)$ into $\Mat_{\ell,n}$. This is because canonical basis for a subspace $V$ is unique.

Next, we prove some important properties of the homomorphism $\phi$ that will be useful in the proof of Theorem \ref{cor:soundness-shortcode-vs-grassmann}. 

First, we show that the map $\phi$ is indeed a homomorphism as promised and thus, preserves edge structure.

\begin{lemma}[$\phi$ is a homomorphism]
For $\phi = \phi_B$ defined above and any $V,V' \in \cV_{\ell}(B)$, $V \sim V'$ in $\grassmann(\ell,n)$ iff $\phi(V) \sim \phi(V')$ in $\cS_{\ell,n}$.
\end{lemma}


\begin{proof}
Let $u \in GF(2)^{\ell}, v \in GF(2)^{n-\ell}$ be arbitrary non-zero vectors that define a rank 1 matrix $uv^{\top}$. Consider the matrix $M =M_V + uv^{\top}$. Then, $M \in Mat_{\ell, n-\ell}$ and thus $\phi^{-1}(M) = W \in \mathcal{V}_{\ell}(B)$. We claim that $dim(W \cap V) = \ell-1.$ Suppose $b_1, b_2, \ldots, b_{\ell}$ are the rows of $M_V$. Then, the rows of $M$ are given by $b_i + u_i v$. Thus, $W$ is spanned by $(e_i, b_i + u_i v)$ where $e_i$ is the $i^{th}$ standard basis element on the first $\ell$ coordinates and the notation $(e_i, b_i + u_i v)$ indicates the concatenation of the vectors in the ordered pair to get a $n$ dimensional vector. In particular, every element of $W$ can be written as $\sum_{i \leq \ell} \lambda_i (e_i, b_i) + (\sum_{i \leq \ell} \lambda_i u_i) v$ and any such vector is contained in $V$ if $(\sum_{i \leq \ell} \lambda_i u_i)$ implying that $dim(V \cap W) = dim(V)-1 = \ell-1.$

On the other hand, let $V'$ be a subspace in $\mathcal{V}_{\ell}(B)$ such that $V' \sim V$ and let $M_{V}$ and $ M_{V'}$ be the matrices obtained via the map $\phi$. Then, $M_{V}$ and $M_{V'}$ must differ in at least one row, say, WLOG, the last row of $M_{V}$ and $M_{V'}$ are $(e_{\ell},v)$ and $(e_{\ell},v')$ respectively. Notice that since the vector with $e_{\ell}$ in the first $\ell$ coordinates is unique in $V,V'$, neither of $(e_{\ell},v), (e_{\ell},v')$ belong to the intersection $V \cap V'$. Further, for every vector $z \in V$, either $z$ or $z + (e_{\ell}, v)$ must be contained in the intersection $V \cap V'$ (as the extra linear equation that $V \cap V'$ satisfies over and above $V$ is satisfied by exactly one of $z$ and $z+(e_{\ell},v)$. Thus, by letting $b'_i =  b_i + (\ell, v) + (\ell, v')$ to every one of the canonical basis elements $b_i$ of $V$ that are not in $V \cap V'$, we get a set of elements that are all 1) contained in $V'$ 2) $Proj_{[\ell]} b'_i = e_i$ for every $i$. This then has to be the canonical basis of $M_{V'}$ (by uniqueness of the canonical basis) and further, the corresponding $M_{V'}$ can be written as $1_S (w+w')^{\top}$ where $S$ is the set of $i$ such that $b_i$ is not in $V \cap V'$.
\end{proof}

Next, we want to argue that expansion of sets is preserved up to constant factors under the map $\phi$. Towards this, we first show that $\cV_{\ell}(B)$ contains a fraction of the vertices of $\grassmann(\ell,n)$ as we next show. 

\begin{lemma}[Projections of Subspaces] 
Let $V \sim \cV_{\ell}$ for $\ell \leq \sqrt{n}/2.$ Then, $\Pr[ \dim \Proj_{[\ell]}(V) = \ell] \geq 0.288$ for large enough $n$ and $\ell = \omega(1).$ 

Further, let $V \in \cV_{\ell}(B)$ for some $B$. Then, at least $1/2$ fraction of the neighbors of $V$ in $\grassmann(\ell,n)$ are contained in $\cV_{\ell}(B)$.  \label{lem:proj-prob}
\end{lemma}

\begin{proof}
We can sample a random subspace of $\ell$ dimension as follows: Choose $\ell$ uniformly random and independent points from $GF(2)^n$. If they are linearly independent, let $V$ be the subspace spanned by them. 

We can estimate the probability that the sampled points are linearly independent as: $\Pi_{i = 0}^{\ell-1} (1-2^{-n+i}) \geq 1-2^{-n} 2^{\ell^2}.$ 

Next, we estimate the probability that the projection to first $\ell$ coordinates of the sampled vectors is linearly independent. By a similar reasoning as above, this probability is at least $\Pi_{i = 0}^{\ell-1} (1-2^{-\ell+i}) \approx 0.289$ (the limit of this product for large $\ell$.)

By a union bound, thus, a random subspace has a full dimensional projection on $S$ with probability at least $0.289-2^{-n/2}$ for any $\ell < \sqrt{n}/2.$ 

For the remaining part, assume that $B = B_0$ - the standard basis. Notice that a random neighbor of $V$ can be sampled as follows:  choose a uniformly random basis for $V$, say $v_1, v_2, \ldots, v_{\ell}$. Replace $v_{\ell}$ by a uniformly random vector $v_{\ell}'$outside of $V$ in $\F_2^n$. Since $V \in \cV_{\ell}(B)$, the projection of $V$ to the first $\ell$ coordinates is linearly independent. $V'$ would thus satisfy the same property whenever $v_{\ell'}$ is such that the projection of $v_{\ell}'$ to the first $\ell$ coordinates is not in the span of the projection to the first $\ell$ coordinates of $v_1, v_2, \ldots, v_{\ell-1}$. The chance of this happening is exactly $1/2$. This completes the proof.
\end{proof}


As a consequence of above, we can now obtain that the preimages of non-expanding sets under $\phi$ are non-expanding in $\grassmann(\ell,n)$.

\begin{lemma} \label{lem:preserve-expansion}
Let $T \subseteq \Mat_{\ell,n}$ be a subset satisfying $\Pr_{M \sim T, M' \sim M} [ M' \in T] = \eta$. Then, $\phi^{-1}(T)$ satisfies: $\Pr_{V \sim \phi^{-1}(T), V' \sim V}[ V' \in \phi^{-1}(T)] \geq \eta/2$.
\end{lemma}

\begin{proof}
Let $B$ the basis used to construct $\phi$. Then, $\phi(T) \subseteq \cV_{\ell}(B)$.
By Lemma \ref{lem:proj-prob}, 1/2 the neighbors of $\phi(T)$ are contained in $\cV_{\ell}(B)$. By assumption, $\eta$ fraction of these neighbors are contained inside $T$. This finishes the proof. 
\end{proof}

Via a similar application of Lemma \ref{lem:proj-prob}, we can establish an appropriate converse.
\begin{lemma} \label{lem:preserve-expansion-converse}
Let $S \subseteq \cV_{\ell}$ be a subset satisfying $\Pr_{V \sim S, V' \sim V}[V' \in S] \geq \eta$. Then, for a uniformly random choice of basis $B$ for $\F_2^n$, $\E_{B} |\phi(S \cap \cV_{\ell}(B)| = \Omega(|S|)$ and $\Pr_{M, M' \sim \phi(S \cap \cV_{\ell}(B)), M' \sim M}[ M' \in \phi(S \cap \cV_{\ell}(B)] \geq \Omega(\eta)$.
\end{lemma}


Finally, we show that $r$-nice sets in $\grassmann{\ell,n}$ get mapped to $r$-nice sets in $\cS(\ell,n)$ and vice-versa.

\begin{lemma}\label{lem:nice-sets-to-nice-sets}
Let $S \subseteq \cV_{\ell}$ be an $r$-nice set in $\grassmann(\ell,n)$. Then, $\phi_B(S \cap \cV_{\ell}(B))$ is an $r$-nice set in $\cS_{\ell,n}$. Conversely, if $T \subseteq \Mat_{\ell,n}$ is an $r$-nice set in $\cS_{\ell,n}$ then $\phi^{-1}(T) = Q \cap \cV_{\ell}(B)$ for some $r$-nice set $Q$ in $\grassmann(\ell,n)$.
\end{lemma}
\begin{proof}
WLOG, assume that $B = B_0$.
We will assume that $S \subseteq \cV_{\ell}$ is the set of all subspaces in $\cV_{\ell}$ contained in a subspace $W$ of co-dimension $r$. The general case is analogous. Equivalently, if $w_1, w_2, \ldots, w_r$ form a basis for $W$, then, for every $V \in S \cap \cV_{\ell}(B)$ and ever $v \in V$ $\langle v, w_i \rangle  =0$ for every $i$. 

Consider the canonical basis $v_1, v_2, \ldots, v_{\ell}$ for $V$ - recall that this means that the projection of $v_i$ to the first $\ell$ coordinates equal $e_i$. Thus, for every $i$, we can write $v_i = (e_i, v'_i)$ for some vectors $v'_i$ of $n-\ell$ dimensions.

Then, $\phi(V)$ is the matrix $M_V$ with rows $v'_i$ by our construction. In particular, this means that the $M_V$ satisfies the constrain: $M_V \cdot w_i = t_i$ where $t_i$ is the vector with $j$th coordinate equal to $\langle e_j, w_i \rangle$. Thus, we have shown that for every $V \in S$, $\phi(V)$ satisfies a set of $r$ affine linear equations. 

Conversely, observe that if any $M$ satisfies the affine linear equation $M_V w_i = t_i$ as above, the set of all $(e_i,u_i)$ for $i \leq \ell$ where $u_i$ is the $i$th row of $M_V$, must span a subspace in $S$. This yields that $\phi(S \cap \cV_{\ell}(B))$ is an $r$-nice set. 

The converse follows from entirely similar ideas. Suppose $T \subseteq \Mat_{\ell,n}$ is an $r$-nice set. WLOG, we restrict to the case where $T$ is the set of all matrices satisfying linear constraints $M q_i = t_i$ for some choice of $r$ linearly independent constraints $(q_i,t_i)$. Letting $u_1, u_2, \ldots, u_{\ell}$ be the rows of $M$, this implies that every vector $v$ in the span of $(e_i,u_i)$ for $i \leq \ell$ satisfies the linear equation $\langle q, v \rangle  =0$ where $q = (q_i,t_i(1), t_i(2), \ldots, t_i(\ell))$. This immediately yields that $\phi^{-1}(M)$ is contained in a subspace $W$ of co-dimension $r$. Conversely, it is easy to check that for every subspace $V$ of dimension $\ell$ contained in $W \cap \cV_{\ell}(B)$, $\phi(V)$ satisfies the $r$ affine linear constraints above. 

This completes the proof.

\end{proof}

\subsection{Shortcode Test vs Grassmann Test}
We now employ the homomorphism constructed in the previous subsection to relate the soundness and expansion hypothesis in shortcode and Grassmann tests. 

First, we show that the soundness hypothesis for degree 2 shortcode consistency test implies the soundness hypothesis for the Grassmann consistency test and complete the proof of Theorem \ref{cor:soundness-shortcode-vs-grassmann}.

\begin{lemma}
The degree 2 shortcode soundness hypothesis (Hypothesis \ref{hypothesis:soundness-unique}) implies the Grassmann soundness hypothesis (Hypothesis \ref{hypothesis:soundness}).
\end{lemma}
\begin{proof}
Let $F$ be the assumed labeling strategy in Hypothesis \ref{hypothesis:soundness}. We will construct a labeling strategy for $\cS_{\ell,n}$ from $G$ so that we can apply the conclusion of \ref{hypothesis:soundness-unique}. We will first choose an embedding of the type we constructed before in order to construct $G$.

Let $B \sim \cB$ be chosen uniformly at random and let $\phi = \phi_B$ as in the previous subsection. For any $V \in \cV_{\ell}(B)$, let $F(V) = f$, a linear function restricted to $V$. Let $v_1, v_2, \ldots, v_{\ell}$ be the canonical basis for $V$, i.e., the projection of $v_i$ to the first $\ell$ coordinates (when written in basis $B$) equals $e_i$ for every $i$. Set $G(\phi(V)) = z$ where $z_i = f(v_i)$. Since $\phi$ is a onto, this defines a labeling strategy for all of $\Mat_{\ell,n}$. 

Next, we claim that if $F$ passes the Grassmann consistency test with probability $\eta$ then $G$ passes the degree 2 shortcode consistency test with probability $\Omega(\eta)$. 

Before going on to the proof of this claim, observe that this completes the proof of the lemma. To see this, we first apply Hypothesis \ref{hypothesis:soundness-unique} to conclude that there's an $r$-nice set $Q$ in $\cS_{\ell,n}$ and an affine function defined by $z \in \F_2^{n-\ell}, u \in \F_2^{\ell}$ such that the labeling strategy $H(M) = Mz+u$ passes the degree 2 shortcode consistency test with probability $\delta$ for all $M$ in $Q$. It it easy to construct the an analogous linear strategy for the Grassmann consistency test: For any $V \in \cV_{\ell}(B)$ with the canonical basis $v_1, v_2, \ldots, v_{\ell}$ defined above, set $f(v_i) = u_i + \langle v_i, z \rangle$. Extend $f$ linearly to the span of all such vectors. Finally, extend $f$ to all vectors by taking any linear extension.  From Lemma \ref{lem:proj-prob}, 1/2 the neighbors of vertices in $\phi^{-1}(Q)$ are contained in $\cV_{\ell}(B)$. From Lemma \ref{lem:nice-sets-map-to-nice-sets}, $\phi^{-1}(Q)= \cal{F} \cap \cV_{\ell}(B)$ for some $r$-nice set $\cal{F}$ in $\grassmann(\ell,n)$. Finally, by an argument similar to the one in Lemma \ref{lem:proj-prob}, $|\cal{F} \cap \cV_{\ell}(B)| \geq \Omega(|\cal{F}|)$ with high probability over the draw of $B$. Combining the above three observations yileds that $f$ passes the Grassmann consistency test when restricted to the nice set $\cal{F}$ with probability $\Omega(\delta)$.


We now complete the proof of the claim. This follows immediately if we show that for any $V \sim V'$ chosen from $\cV_{\ell}(B)$, $\Pr_{V \sim V', V,V' \in \cV_{\ell}(B)} [F(V)_{\mid V} = F(V')_{\mid V'}] \geq 0.07 (\eta - 2^{-n+\ell}).$ 

Without loss of generality, we assume that $B$ is the standard basis $\{e_1, e_2, \ldots,e_n\}$. 
First, notice that $ \Span \{ V \cup V'\}$ is of dimension $\ell+1$ for  for all but $2^{-n+\ell}$ fraction of pairs $V \sim V'$. Thus, we can assume that $\Pr_{V \sim V' \mid \dim \Span \{V \cup V'\} = \ell+1} [ F(V)_{\mid V} = F(V')_{\mid V'}] \geq \eta - 2^{-n+\ell}.$

Let $C = B^{-1}$, the basis change matrix corresponding to $B$ and let $C_i$ be the $i^{th}$ row of $C$ and let $C_{[\ell]}$ be the matrix formed by taking the first $\ell$ rows of $C$.
Fix $V \sim V'$ for some $V,V' \in \cV_{\ell}$.  Assume now that $\Span \{ V \cup V' \}$ is of dimension $\ell+1$. Let $v_1, v_2, \ldots, v_{\ell-1}$ be a basis for $V \cap V'$. Let $V = \{ V \cap V' \cup w_1\}$ and $V' = \{V \cap V' \cup w_2 \}$ for some $w_1, w_2$ that linearly independent of each other and of any vector in $V \cap V'$. We estimate the probability that $V,V' \in \cV_{\ell}(B)$.  Then, this is the probability that $v_1, v_2, \ldots, v_{\ell-1}, w_1, w_2$ are mapped by $C^{[\ell]}$ into $a_1, a_2, \ldots,a_{\ell-1}, a_{\ell}, a_{\ell+1}$ respectively, satisfying $a_{\ell} , a_{\ell+1} \not \in \Span \{a_i \mid i\leq \ell-1\}$. It is easy to check that the probability of this over the random choice of $B$ is at least $0.288*1/4 > 0.07$. This proves the claim. 

By taking $n$ large enough (compared to $\ell$), this probability can be made larger than, say, $0.06 \eta$ (say). This finishes the proof.

\end{proof}

Next, we show that the Grassmann Expansion Hypothesis (Hypothesis \ref{hypothesis:expansion}) is equivalent to the Inverse Shortcode Hypothesis (Hypothesis \ref{conj:inv-short-code}) and complete the proof of Theorem \ref{thm:expansion-equivalent}.

\begin{lemma}
The Grassmann Expansion Hypothesis (Hypothesis \ref{hypothesis:expansion}) is equivalent to the Inverse Shortcode Hypothesis (Hypothesis \ref{conj:inv-short-code}).
\end{lemma}
\begin{proof}
First, we show that Hypothesis \ref{hypothesis:expansion} implies Hypothesis \ref{conj:inv-short-code}.

Let $S \subseteq \Mat_{\ell,n}$ be such that $\Pr_{M \sim S, a \in \F_2^{\ell}, b \in \F_2^{n}} [ M+ ab^{\top} \in S] = \eta$. 
Then, by Lemma \ref{lem:preserve-expansion}, $\phi_B^{-1}(S)$ has an expansion of $\Omega(\eta)$ in $\grassmann(\ell,n)$. 

Applying the Grassmann expansion hypothesis (Hypothesis \ref{hypothesis:expansion}), we know that there exists a $r$-nice set $\cal{F}$ in $\grassmann(\ell,n)$ such that $|\cal{F} \cap \phi_B^{-1}(S)| \geq \delta|\cal{F}|$. Further, since $\phi_B^{-1}(S) \subseteq \cV_{\ell}(B)$, we must have: $|(\cal{F} \cap \cV_{\ell}(B)) \cap \phi_B^{-1}(S)| \geq \delta | \cal{F} \cap \phi_B^{-1}(S)|$. To finish, observe that by Lemma \ref{lem:nice-sets-to-nice-sets}, $\phi(\cal{F} \cap \phi_B^{-1}(S))$ is an $r$-nice set, say $\cal{Q}$ in $\cS_{\ell,n}$. This, show that $|S \cap \cal{Q}| \geq \delta |\cal{Q}|$ completing the proof.

The proof of the other direction, that is, Hypothesis \ref{conj:inv-short-code} implies Hypothesis \ref{hypothesis:expansion}, is analogous and relies on the use of Lemma \ref{lem:preserve-expansion-converse}.
\end{proof}

\addreferencesection
\bibliographystyle{amsalpha}
\bibliography{bib/refs}

\end{document}